\documentclass[10pt]{sigplanconf}

\usepackage{mathptmx}
\DeclareSymbolFont{epsilon}{OML}{ntxmi}{m}{it}
\DeclareMathSymbol{\epsilon}{\mathord}{epsilon}{"0F}

\renewcommand*{\bibfont}{\small}

\usepackage{flushend}
\usepackage{courier}            
\usepackage[scaled]{helvet} 
\usepackage{url}                  
\usepackage{listings}          
\usepackage{enumitem}      
\usepackage[hidelinks=true,breaklinks,draft=false]{hyperref}   

\usepackage{pslatex}
\usepackage{algpseudocode}
\usepackage{algorithmicx}
\usepackage{stmaryrd}
\usepackage{xspace}

\usepackage{xspace}
\usepackage{graphicx}
\usepackage{proof}
\usepackage{amsmath}
\usepackage{amsthm}
\usepackage{amssymb}
\usepackage{float}

\usepackage[T1]{fontenc}
\usepackage[kerning,spacing]{microtype}
 
\newcommand{\hide}[1]{} 
\newcommand{\C}[1]{\lstinline!#1!}

\newcommand{\ie}{\textit{i.e.}}

\newtheorem{thm}{Theorem}
\newtheorem{lem}{Lemma}

\numberwithin{equation}{section}

\usepackage{listings}
\lstdefinelanguage{jelf}{
  morekeywords = {
       bool, true, false
     , filter, null
     , concretize, create
     , do
     , else, empty, error, exn
     , facet, flowsto, for, from, function
     , get
     , high
     , if, in, input, int, float, bool, string
     , join
     , label, let, length, list, low
     , map, mkLabel, mkSensitive
     , nil, null, none
     , option, or, orderby
     , policy, print, public
     , record, ref, respond, restrict
     , save, schema, select, send_mail
     , then, this
     , unit
     , where, with },
  morecomment=[s]{/*}{*/},
  literate=
    {==}{{$=\!\!=$} }1
}
\usepackage{color}
\definecolor{dkgreen}{rgb}{0,0.3,0}
\definecolor{gray}{rgb}{0.5,0.5,0.5}
\definecolor{mauve}{rgb}{0.58,0,0.82}
 
\usepackage{url}

\renewcommand{\t}[1]{\text{#1}}

\renewcommand{\vss}{\vspace{10pt}}

\newcommand{\ors}{\ensuremath{\ |\ \ }}

\newcommand{\code}[1]{\lstinline[language=jelf,columns=fullflexible,basicstyle=\fontfamily{lmss}\selectfont\normalsize]{#1}}  
\newcommand{\python}[1]{\lstinline[language=Python,columns=fullflexible,basicstyle=\fontfamily{lmss}\selectfont\normalsize]{#1}}
\newcommand{\sql}[1]{\lstinline[language=SQL,columns=fullflexible,basicstyle=\fontfamily{lmss}\selectfont\normalsize]{#1}}

\newcommand{\defeq}{\stackrel{\mathrm{def}}{=}}
\newcommand{\lam}[2]{\ensuremath{\lambda #1 .\hspace{0.01em} #2}}
\newcommand{\labexpr}[3]{\langle {#1}~?~{#2}:{#3} \rangle}
\newcommand{\fun}[1]{{\mathit {#1}}}

\newcommand{\pc}{\fun{pc}}

\newcommand{\uv}{V}

\newcommand{\UValue}{\fun{Val}}
\newcommand{\URawValue}{\fun{RawValue}}
\newcommand{\ULabel}{\fun{Label}}
\newcommand{\USignedLabel}{\fun{Branch}}
\newcommand{\USignedLabels}{\fun{Branches}}

\newcommand{\UString}{\fun{String}}

\newcommand{\labval}[3]{\langle {#1}~?~{#2}:{#3} \rangle}
\newcommand{\labResult}[3]{\langle\!\langle\, {#1}~?~{#2}:{#3} \,\rangle\!\rangle}

\newcommand{\lamc}[2]{(\lambda {#1}.{#2})}
\newcommand{\sstep}[5]{{#1},{#2} \Downarrow_{#3} {#4},{#5}}

\newcommand{\stmtstep}[4]{{#1},{#2} \Downarrow {#3},{#4}}
\newcommand{\mrule}[3]{
  \multicolumn{1}{l}{\rel{#1}}\\
  \frac{\strut\begin{array}{@{}c@{}} #2 \end{array}}
       {\strut\begin{array}{@{}c@{}} #3 \end{array}}
   \\~\\
}
\newcommand{\mmrule}[3]{
  \frac{\strut\begin{array}{@{}c@{}} #2 \end{array}}
       {\strut\begin{array}{@{}c@{}} #3 \end{array}}
  &~~{\rel{#1}}
   \\~\\
}

\newcommand{\notlab}[1]{\neg #1}

\newcommand{\joinLab}[2]{#2 \cup\{#1\} }

\newcommand{\jtable}[1]{\t{table} \ #1}
\newcommand{\jrow}[1]{\overline{#1}}

\newcommand{\visible}[2]{#1 \sim #2}
\newcommand{\nvisible}[2]{#1 \not\sim #2}

\newcommand{\setuprulecol}{\hspace{0.5in} \\[-.2in]}

\newcommand{\commentOut}[1]{}

\newcommand{\equivlab}[1]{\sim_{#1}}
\newcommand{\ctxthole}{~\bullet~}

\usepackage{tikz}
\usetikzlibrary{calc}
\usetikzlibrary{shapes}
\usetikzlibrary{arrows}

\newcommand{\dblang}{$\lambda^{JDB}$\xspace}
\newcommand{\webframework}{Jacqueline\xspace}

\newcommand{\corelang}{\mbox{$\lambda^{\rm\emph{jeeves}}$}\xspace}

\input{brackets.sty}

\begin{document}
\toappear{}

\setlength{\abovedisplayskip}{0pt}
\setlength{\belowdisplayskip}{0pt}
\setlength{\abovedisplayshortskip}{0pt}
\setlength{\belowdisplayshortskip}{0pt}
%
%
\title{Precise, Dynamic Information Flow for Database-Backed Applications}
\authorinfo{Jean Yang}{Carnegie Mellon University and Harvard Medical School, USA \vspace{-1ex}}{\vspace{-2ex}}
\authorinfo{Travis Hance}{Dropbox, USA \vspace{-1ex}}{\vspace{-2ex}}
\authorinfo{Thomas H. Austin}{San Jose State University, USA \vspace{-1ex}}{\vspace{-2ex}}
\authorinfo{Armando Solar-Lezama}{Massachusetts Institute of Technology, USA \vspace{-1ex}}{\vspace{-2ex}}
\authorinfo{Cormac Flanagan}{University of California, Santa Cruz, USA \vspace{-1ex}}{\vspace{-2ex}}
\authorinfo{Stephen Chong}{Harvard University, USA \vspace{-1ex}}{\vspace{-2ex}}

\maketitle
\vspace {-0.5ex}
\begin{abstract}
We present an approach for dynamic information flow control across the application and database. Our approach reduces the amount of policy code required, yields formal guarantees across the application and database, works with existing relational database implementations, and scales for realistic applications. In this paper, we present a programming model that factors out information flow policies from application code and database queries, a dynamic semantics for the underlying \dblang core language, and proofs of termination-insensitive non-interference and policy compliance for the semantics. We implement these ideas in \webframework, a Python web framework, and demonstrate feasibility through three application case studies: a course manager, a health record system, and a conference management system used to run an academic workshop. We show that in comparison to traditional applications with hand-coded policy checks, \webframework applications have 1) a smaller trusted computing base, 2) fewer lines of policy code, and 2) reasonable, often negligible, overheads.
\end{abstract}

\category{D.3.3}{Programming Languages}{Language Constructs and Features}

\terms Frameworks, Security

\keywords Web frameworks, information flow

\section{Introduction}
From social networks to electronic health record systems, programs increasingly process sensitive data. As information leaks often arise from programmer error, a promising way to reduce leaks is to reduce opportunities for programmer error.

A major challenge in securing web applications involves reasoning about the flow of sensitive data across the application and database. According to the OWASP report~\cite{SeLINQ}, errors frequently occur at component boundaries. Indeed, the difficulty of reasoning about how sensitive data flows through both application code and database queries has led to leaks in systems from the HotCRP conference management system~\cite{HotCRPBug} to the social networking site Facebook~\cite{FacebookLeak}. The patch for the recent HotCRP bug involves policy checks across application code and database queries.

Information flow control is important to securing the application-database boundary~\cite{DavisC10, SeLINQ, IFDataSecurity, UrWeb}. This is because leaks often involve the results of computations on sensitive values, rather than sensitive values themselves. To reduce the opportunity for inadvertent leaks, we present a \emph{policy-agnostic} approach~\cite{JeevesPaper, AYPLAS13}. Using this approach, the programmer factors out the implementation of information flow policies from application code and database queries. The system manages the policies, removing the need to trust the remaining code. The program thus specifies each policy once, rather than as repeated intertwined checks across the program. Because of this, policy-agnostic programs require less policy code.  We illustrate these differences in Figure~\ref{fig:web_server_comparison}.

\begin{figure}
\centering
\includegraphics[width=3.3in]{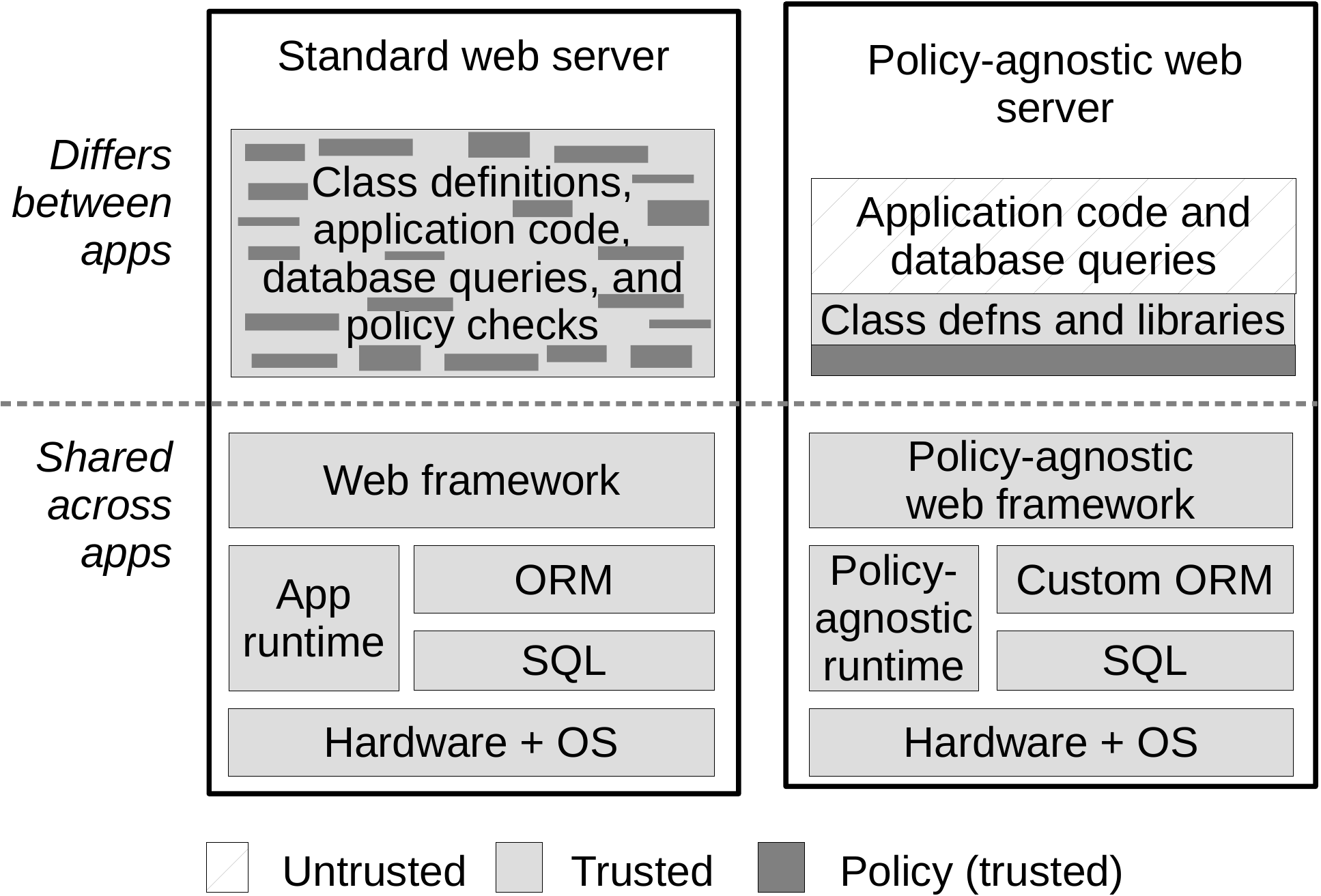}
\caption{Application architecture in a standard web server compared to a policy-agnostic web server.}
\label{fig:web_server_comparison}
\end{figure}

Supporting policy-agnostic programming for web applications requires the framework to enforce information flow policies across the application and database. As we also show in Figure~\ref{fig:web_server_comparison}, a standard web program runs using an application runtime and a database. An object-relational mapping (ORM) to mediate interactions between the two. Our web framework uses a policy-agnostic application runtime and a specialized ORM that mediates interactions between policy-agnostic application code and policy-agnostic database queries.

There are three main parts to our solution: 1) supporting policy-agnostic database queries, 2) providing formal guarantees across the application and database, and 3) addressing issues of practical feasibility. We extend prior work on the Jeeves programming language~\cite{JeevesPaper, AYPLAS13} that defines a policy-agnostic semantics for a simple imperative language. As is common with language-based approaches, Jeeves's guarantees extend only within the Jeeves runtime. Interoperation with external databases is important as web applications rely on commodity databases for performance reasons. The challenge is, then, to support policy-agnostic programming for database queries in a way that leverages existing database implementations while providing strong guarantees.

We present \emph{faceted databases} for supporting policy-agnostic database queries. The Jeeves runtime performs different computations based on the permissions of the user viewing the output. Because the viewer may not be known in advance, the runtime uses \emph{faceted execution} to simulate simultaneous executions. A \emph{faceted value} is the runtime representation of a value that may differ across executions. Semantically, a faceted database stores faceted values and performs faceted query execution. We show how to use a \emph{faceted object-relational mapping} (FORM) to embed faceted values using relational databases and, surprisingly, to support faceted query execution simply by manipulating meta-data. The FORM manages complex dependencies, allowing a policy to query the data it protects.

Next we show that interoperation with faceted databases yields strong guarantees. We extend Jeeves's core language with relational operators to create the \dblang core language. We present a dynamic faceted execution semantics for \dblang and prove termination-insensitive non-interference and policy compliance. The formalization corresponds closely to an implementation strategy using existing database implementations while yielding concise proofs.

Towards supporting realistic applications, we formulate an ``Early Pruning'' optimization. While simulating multiple executions is desirable for reasoning, exploring multiple executions can be expensive in practice. The Early Pruning optimization allows the program to use program assumptions to safely explore fewer executions. This optimization is particularly useful for web applications, where it is often possible to use the session user to predict the viewer. With Early Pruning, performance may even be better than with hand-coded checks, as the runtime may now check policies once rather than repeatedly throughout execution.

Finally, we demonstrate practically feasibility. We present Jacqueline, a web framework based on Python's Django~\cite{Django} framework. We use \webframework to build several application case studies, including a conference management system that we have deployed to run an academic workshop. The case studies show that using \webframework, policies are localized and the size of the policy code is smaller. Consequently, security audits can focus on the localized policy specifications rather than having to review the entire code base. We also demonstrate that \webframework has reasonable, often negligible, overheads. For one case, the \webframework implementation performs better than an implementation with hand-coded policies.


In summary, we make the following contributions:
\begin{itemize}
\item \textbf{Policy-agnostic web programming.} We present an approach that allows programmers to factor out information flow policies from the rest of web programs and rely on a web framework to dynamically enforce the policies.
\item \textbf{Faceted databases.} We present faceted databases to support policy-agnostic relational database queries. We present a faceted object-relational mapping (FORM) strategy for implementing faceted databases using existing relational database implementations.
\item \textbf{Faceted execution for database-backed applications.} We show interoperation of faceted databases with faceted application runtimes by presenting a dynamic semantics for the \dblang core language and proving termination-insensitive non-interference and policy compliance.
\item \textbf{Early Pruning optimization.} We address performance issues by formalizing an optimization, proving that it preserves policy compliance, and demonstrating that it significantly decreases overheads.
\item \textbf{Demonstration of practical feasibility.} We present the \webframework web framework and demonstrate expressiveness and performance through several application case studies. We compare against hand-implemented policies, showing that not only does Jacqueline reduce lines of policy code, but also that policy enforcement has reasonable, often negligible, overheads.
\end{itemize}
Our approach decreases the opportunity for programmer error, provides strong formal guarantees, and is practically feasible.

\section{Introductory Example}
\label{sec:introductory_example}
Using our policy-agnostic web framework, the programmer implements each information flow policy once, associated with the data schemas, as opposed to repeatedly across the code base. We designed \webframework so that programming with it is as similar as possible to programming with Django. In \webframework, the application runtime and object-relational mapping dynamically manipulate sensitive values and policies so the programmer may omit repeated checks.

Consider a social calendar application. Suppose Alice and Bob want to plan a surprise party for Carol, 7pm next Tuesday at Schloss Dagstuhl. They should be able to create an event such that information is visible only to guests. Carol should see that she has an event 7pm next Tuesday, but not that it is a party. Everyone else may see that there is a private event at Schloss Dagstuhl, but not event details.




We demonstrate how to implement this example using \webframework, our new web framework based on Django~\cite{Django}, a \emph{model-view-controller} framework. In a standard MVC framework, the \emph{model} describes the data, the \emph{view} describes frontend page rendering, and the \emph{controller} implements other functionality. An object-relational mapping (ORM) supports a uniform object representation. In \webframework, the model additionally specifies information flow policies. The faceted object-relational mapping (FORM) additionally supports a uniform representation of sensitive values and policies. \webframework is policy-agnostic: other than the policies, a \webframework program looks like a policy-free Django program.

The division of labor between the programmer and the framework is as follows. The programmer associates information flow policies with fields in the data schema, codes within the subset of Python supported by our Jeeves library, and accesses the database only through the \webframework API. The framework tracks sensitive values and policies between the application and database to produce outputs that adhere to the policies. In our attack model, the user is untrusted and we assume the programmer is not malicious.

We intend for this example to explain the semantics of policy-agnostic web programming. We discuss issues of implementation and optimization issues in later sections.

 
\subsection{Schemas and Policies in \webframework{}}
\begin{figure}
\begin{lstlisting}[language=Python, xleftmargin=10pt, numbers=left, escapeinside={~}{~}]
class Event(JModel):
~\label{lst:calendar_name}~  name = CharField(max_length=256)
  location = CharField(max_length=512)
  time = DateTimeField()
~\label{lst:calendar_end_of_django}~  description = CharField(max_length=1024)

  # Public value for name field.
  @staticmethod
~\label{lst:calendar_public}~  def jacqueline_get_public_name(event):
    return "Private event"

  # Public value for location field.
  @staticmethod
~\label{lst:calendar_location_public}~  def jacqueline_get_public_location(event):
    return "Undisclosed location"

  # Policies for name and location fields.
  @staticmethod
  @label_for('name', 'location')
  @jacqueline
~\label{lst:calendar_policy}~  def jacqueline_restrict_event(event, ctxt):
    return (EventGuest.objects.get(
      event=self, guest=ctxt) != None)

~\label{lst:calendar_EventGuest}~class EventGuest(JModel):
  event = ForeignKey(Event)
  guest = ForeignKey(UserProfile)
\end{lstlisting}
\caption{\webframework{} schema fragment for calendar events.}
\label{fig:calendar_code}
\end{figure}

In \webframework's policy-agnostic programming model, programmers are responsible for specifying information flow policies and the application runtime and object-relational mapping are responsible for tracking the flow of sensitive values to produce outputs adhering to those policies. Programmers specify each information flow policy once, associated with the \emph{data schema} in the model. We show a sample schema for the \python{Event} and \python{EventGuest} data objects in Figure~\ref{fig:calendar_code}. A \webframework schema defines field names, field types, and optional policies. We define the \python{Event} class with fields \python{name}, \python{location}, \python{time}, and \python{description}. Up to line~\ref{lst:calendar_end_of_django}, this looks like a standard Django schema definition.

\subsubsection{Secret Values and Public Values}
A sensitive value in \webframework encapsulates a secret (high-confidentiality) view available only to viewers with sufficient permissions and a public (low-confidentiality) view available to other viewers. \webframework allows sensitive values to behave as either the secret value or public value, depending on viewing context (\ie{} the user viewing a page).

The actual field value is the secret view and the programmer must additionally define a method computing the public view. On line~\ref{lst:calendar_public} we define the \python{jacqueline_get_public_name} method computing the public view of the \code{name} field. If the permissions prohibit a viewer from seeing the sensitive \python{name} field, then the \python{name} field will behave as \python{"Private event"} throughout all computations, including database queries. This function takes the current row object (\python{event}) as an argument, allowing public values to be computed using row fields. The \webframework{} ORM uses naming conventions (\ie{} the \python{jacqueline_get_public} prefix) to find the appropriate methods to compute public views.

\subsubsection{Specifying Policies}
In \webframework, programmer-specified information flow policies guard the flow of sensitive values. On line~\ref{lst:calendar_policy} we implement the policy for the fields \python{name} and \python{location}, as indicated by the \python{label_for} decorator. The policy is a method that takes two arguments, the current row object (\python{event}) and the viewer (\python{ctxt}) corresponding to the user looking at a page. Our policy queries the \python{EventGuest} table (line~\ref{lst:calendar_EventGuest}) to determine whether the viewer is associated with the event.

Without \webframework, the programmer would need to implement an equivalent function and call it whenever the \python{location} value is used. Using \webframework, the program no longer needs to explicitly perform these policy checks because \webframework's ORM and application runtime ensure that the policy is enforced.
\webframework{} handles mutable state by enforcing this policy with respect to the value of \python{event} at the time a value is created and the state of the system at the time of output.

\subsection{Faceted Execution}
\webframework uses an enhanced application runtime that keeps track of the secret and public views of sensitive values and results of computations on sensitive values. Once the programmer associates policies with sensitive data fields, the rest of the program may be \emph{policy-agnostic}. We call \python{create} in \webframework the same way as in Django:
\begin{lstlisting}[language=Python]
carolParty = Event.objects.create(
    name = "Carol's surprise party"
  , location = "Schloss Dagstuhl", ...)
\end{lstlisting}
To manage the policies, the \webframework FORM creates faceted values for the sensitive fields. For the \python{name} fields, the framework creates the faceted value $\labexpr{\code{k}}{\code{"Carol's surprise party"}}{\code{"Private event"}}$, where \code{k} is a fresh Boolean \emph{label} guarding the secret actual field value and the public facet computed from the \python{get_public_name} method. The runtime eventually assigns label values based on policies and the viewer. We describe in Section~\ref{sec:solution_overview} how the FORM stores faceted values in a relational database.

The runtime evaluates faceted values by evaluating each of the facets. Evaluating \code{"Alice's events: " + str(alice.events)} yields the resulting faceted value guarded by the same label: 
\begin{lstlisting}[language=Python,columns=fixed]
$\langle$k ? "Alice's events: Carol's surprise party"
   : "Alice's events: Private event"$\rangle$
\end{lstlisting}
Guests of the event will see \code{"Carol's surprise party"} as part of the list of Alice's events, while others will see only \code{"Private event"}. Faceted execution propagates labels through all derived values, conditionals, and variable assignments to prevent indirect and implicit flows.

\webframework{} performs faceted execution for database queries, preventing indirect flows through queries like the following:
\begin{lstlisting}[language=Python]
Event.objects.filter(
  location="Schloss Dagstuhl")
\end{lstlisting}
If \python{carolParty} is the only event in the database, faceted execution of the \python{filter} query yields a faceted list $\labexpr{\python{m}}{\python{[carolParty]}}{\python{[]}}$. Viewers who should not be able to see the \python{location} field will not be able to see values derived from the sensitive field.

\webframework{} also prevents implicit leaks through writes to the database. For instance, consider this code that replaces the \python{description} field of \python{Event} rows with \python{"Dagstuhl event!"} when the \python{location} field is \python{"Schloss Dagstuhl"}:
\begin{lstlisting}[language=Python]
for loc in Event.objects.all():
  if loc.location == "Schloss Dagstuhl":
    loc.description = "Dagstuhl event!"
    save(loc)
\end{lstlisting}
For \python{carolParty} the condition evaluates to $\labexpr{\python{k}}{\python{True}}{\python{False}}$. The runtime records the influence of \python{k} when evaluating the conditional so that the call to \python{save} writes $\labexpr{\python{k}}{\python{carolPartyNew}}{\python{carolParty}}$, where \python{carolPartyNew} is the updated value.

\subsection{Computing Concrete Views}
\label{sec:concrete_views}
Computation sinks such as \python{print} take an additional argument corresponding to the viewer and resolves policies according to the viewer and policies. For instance, \python{print carolParty.name} displays \python{"Carol's surprise party"} to some viewers and \python{"Private event"} to others. The programmer does not need to designate the viewer: it can be an implicit parameter set from authorization information.

The policies and viewer define a system of constraints for determining label values. Printing \code{carolParty.name} to \code{alice} corresponding to the following constraint:
\begin{lstlisting}[language=Python]
k $\Rightarrow$
  (EventGuest.objects.get(
    event=self, guest=ctxt) != None)
\end{lstlisting}
To account for dependencies on mutable state, the runtime evaluates this constraint in terms of the guest list at the time of output. Labels are the only free variables in the fully evaluated constraints. There is always a consistent assignment to the labels: assigning all labels to \python{False} is always valid.

The constraint semantics allows \webframework to handle mutual dependencies between policies and sensitive values. Suppose that the guest list policy depended on the list itself:
\begin{lstlisting}[language=Python]
@label_for('guest')
def jacqueline_restrict_guest(eventguest, ctxt):
  return (EventGuest.objects.get(
    event=eventguest.e, guest=ctxt) != None)
\end{lstlisting}
The policy requires that there must be an entry in the \python{EventGuest} table where the \python{guest} field is the viewer \python{ctxt}, so the policy for the \python{guest} field depends on the value of the field itself. There are two valid outcomes for a viewer who has access: either the system shows empty fields or the system shows the actual fields. \webframework{} always attempts to show values unless policies require otherwise. Note that unless there are mutual dependencies, \webframework may determine label values by evaluating policies directly.

Such circular dependencies are increasingly common in real-world applications. Consider, for instance, the following policies: a viewer must be within some radius of a secret location to see the location; a viewer must be a member of a secret list to see the list. Unfortunately, it is common practice to execute such policies in a trusted ``omniscient'' context that risks leaking information.


\section{The Faceted Object-Relational Mapping}
\label{sec:solution_overview}
Our faceted object-relational mapping (FORM) 1) uses meta-data to represent faceted values and 2) manages queries by manipulating meta-data and marshalling to and from the database representation. Surprisingly, our solution allows us to use existing relational database implementations for creating, updating, selecting, joining, and sorting records.
In this section, we introduce the faceted object-relational mapping (FORM) using SQL syntax and present the Early Pruning optimization.

\subsection{Executing Relational Queries with Facets}
\label{sec:sql_mapping}

\begin{table*}
\centering
%
\begin{tabular}{lllll}
\sql{id} & \sql{name} & \sql{location} & \sql{jid} & \sql{jvars}\\
\hline
1 & \sql{"Carol's ... party"} &\sql{"Schloss Dagstuhl"} & 1 & \sql{"x=True"}\\
2 & \sql{"Private event"} & \sql{"Undisclosed location"} & 1 & \sql{"x=False"}\\
\end{tabular}\\
\caption{Example table.}
\label{tab:comparison_vanilla_jelf}
\end{table*}

\begin{table*}
\centering
\begin{tabular}{ll}
\textbf{Django Query} & \textbf{\webframework{} Query}\\
\hline
\hline
\noalign{\vskip 2mm}
	\multicolumn{2}{c}{ \python{EventGuest.objects.filter(guest__name="Alice")} }\\
\begin{lstlisting}[language=sql, escapeinside={!}{!}]
SELECT EventGuest.event, EventGuest.guest
  FROM EventGuest
  JOIN UserProfile
    ON EventGuest.guest_id = UserProfile.id
  WHERE UserProfile.name='Alice';

  !~!
\end{lstlisting}
	&
\begin{lstlisting}[language=sql]
SELECT EventGuest.event, EventGuest.guest,
    EventGuest.jid, EventGuest.jvars,
    UserProfile.jvars
  FROM EventGuest
  JOIN UserProfile
    ON EventGuest.guest_id = UserProfile.jid
  WHERE UserProfile.name='Alice';
\end{lstlisting}
\end{tabular}
\caption{Translated ORM queries in Django vs. \webframework{}.}
\label{tab:comparison_vanilla_jelf_sql}
\end{table*}

A \emph{faceted row} is a faceted value containing leaves that are non-faceted relational records. Any record containing faceted values may be rewritten to be of this form. We map each faceted row to multiple database rows by augmenting records with meta-data columns corresponding to 1) a unique identifier \sql{jid} and 2) an identifier \sql{jvars} describing which facet the row corresponds to, for instance \sql{"k1=True,k2=True"}.

The FORM is responsible for marshalling between the database and runtime representations of faceted values. The FORM stores the faceted value $\labexpr{\code{k}}{\code{"Carol's surprise party"}}{\code{"Private event"}}$ as two rows in the \sql{Event} table with the same \sql{jid} of \sql{1}. The secret facet has a \sql{jvars} value of \sql{"k=True"} and the public facet has a \sql{jvars} value of \sql{"k=False"}.
For nested facets, we store more labels in the \sql{jvars} column, for instance \sql{"k1=True,k2=True"}.
In Table~\ref{tab:comparison_vanilla_jelf} we show how this faceted value would look in an augmented table.

\subsubsection{Queries That Track Sensitive Values}
A key advantage of our representation is that the FORM can issue standard relational queries not only for selections and projections, but also joins and sorts. Storing each facet in a different row allows the FORM to rely on the correct marshalling of query results for preventing indirect flows through queries. Note that the FORM would not be able to issue relational queries in such a straightforward way, for instance, if it stored each faceted value in the same row, or if it stored different facets in different databases.

Consider the query \sql{SELECT *} \sql{from Event WHERE} \sql{location = "Schloss Dagstuhl"} on the rows from Table~\ref{tab:comparison_vanilla_jelf}. Issuing the query directly on the augmented database will return the one matching row with \sql{jid=1} and \sql{jvars="k=True"}.
Reconstructing the facet structure yields a faceted value guarded by label \sql{k} with a collection containing the record in the secret facet and an empty collection in the other facet. Relying on unmarshalling is sufficient for faceted execution.

Surprisingly, rows from joins that occur based on sensitive values will also be appropriately guarded by the appropriate path conditions. The only additional considerations the FORM needs to make for joins are to 1) take into account the \sql{jvars} fields from both tables and 2) ensure that foreign keys (references into another table) use \sql{jid} rather than the primary key. In Table \ref{tab:comparison_vanilla_jelf_sql}, we show an example where the \sql{WHERE} clause filters on the results of a \sql{JOIN}. In the \sql{ON} clause, we use the \code{jid} rather than \code{id}. In the \sql{SELECT} clause, we include the \code{User.jvars} as well as the \code{EventGuest.jvars} field.

A particularly nice consequence of storing each facet in different rows is that the FORM can take advantage of SQL's \sql{ORDER BY} functionality for sorting. Suppose we had faceted records, each with a single field \code{f}, with values $\labexpr{\code{a}}{\code{"Charlie"}}{\code{"***"}}$, $\labexpr{\code{b}}{\code{"Bob"}}{\code{"***"}}$, and $\labexpr{\code{c}}{\code{"Alice"}}{\code{"***"}}$.
The FORM can use the standard sorting procedure without leaking information because the secret values are stored in different rows from the public values. Correct unmarshalling will enforce the policies so that, for instance, an output context with the permitted labels $\{a, \notlab b, c \}$ would see \code{["***", "Alice", "Charlie"]}.

A limitation is that the FORM cannot use existing relational implementations for aggregation, for instance counting or summing. Using aggregate queries directly could leak information because without looking at the path conditions, these aggregates would combine values across facets. This does not suggest a fundamental limitation. Applications often prematerialize aggregates, making it reasonable to use the faceted runtime to precompute aggregates. Otherwise, supporting faceted aggregation at scale is a matter of optimizing the procedures, perhaps as database user-defined functions.

\subsubsection{Creating and Updating Data and Policies}
The FORM creates tables and rows with the appropriate meta-data to keep track of facets. The FORM prevents implicit leaks through updates by updating meta-data appropriately and potentially deleting rows. Invoking \python{save} in branches that depend on faceted values creates facets that incorporate the path conditions.
To add policies, the programmer needs to manipulate only the meta-data columns (\python{jvars} and \python{jid}). Adding policies to legacy data involves adding meta-data columns. Updating policies using existing labels simply involves updating policy code.

\subsection{Early Pruning Optimization}
\label{sec:early_pruning}

An important correctness-preserving optimization is to prune facets once the runtime knows the viewer. This involves being able to determine 1) the viewing context and 2) that policy-relevant state relevant will not change before output. Two properties of web programs make this analysis simple. First, the session user is often the viewing context. Second, computation sinks are easy to identify in model-view-controller frameworks: most functions either read from the database or write to the database, but not both.
This makes it advantageous for the framework to speculate on the viewer for ``get'' requests.
We formalize Early Pruning in Section~\ref{sec:early_pruning_rule}.

\subsection{Data Representation Considerations}
\label{sec:data_considerations}
It is also important to discuss whether storing faceted values in the database may be prohibitively expensive. There are many ways to avoid storing too much data in practice. Work on multi-level databases~\cite{DenningAMNSH86, Lunt1990} suggests it is both useful and practically feasible to store multiple versions of data corresponding to different access levels. The question becomes, then, how to avoid storing too much data due to too many possible path conditions. An important optimization involves combining values that are the same to a single view. In Section~\ref{sec:formal_semantics}, we define an optimization to allow sharing rows that different facets have in common.

\newcommand{\ctxteval}[2]{#1[#2]}

\section{Formal Semantics and Policy Compliance}
\label{sec:formal_semantics}

We model the faceted object relational mapping with the idealized core language called \dblang. We prove that \dblang satisfies termination-insensitive non-interference and policy compliance across the application and database.

\subsection{Syntax and Formal Semantics}
\newcommand{\ue}{e}
\newcommand{\mkref}[1]{\t{ref}~#1}
\newcommand{\deref}[1]{\t{!}#1}
\newcommand{\assign}[2]{#1\t{:=}\,#2}

\newcommand{\readFile}[1]{\t{read}(#1)}
\newcommand{\writeFile}[2]{\t{write}(#1,#2)}
\newcommand{\print}[2]{\t{print}~\{#1\}~#2}
\newcommand{\levelDecl}[2]{\t{label}~{#1}~\t{in}~{#2}}

\newcommand{\fh}{f}

\newcommand{\defaultVal}{0}

\newcommand{\ife}[3]{\t{if}~{#1}~\t{then}~{#2}~\t{else}~{#3}}
\newcommand{\letin}[3]{\t{let}~{#1}~\t{=}~{#2}~\t{in}~{#3}}

\newcommand{\true}{\fun{true}}
\newcommand{\false}{\fun{false}}

\newcommand{\policy}[2]{\t{restrict}(#1,#2)}
\newcommand{\policyStd}[1]{\t{restrict}(#1)}
\newcommand{\policyVar}{policy}
\newcommand{\policyStore}[1]{#1^{policy}}
\newcommand{\policyAdd}{\wedge_f}

\newcommand{\fileOut}[2]{{#1}\!:\!{#2}}

\newcommand{\Expr}{Expr}
\newcommand{\row}[1]{\t{row} \ #1}
\newcommand{\select}[3]{\sigma_{#1=#2} \ #3}
\newcommand{\project}[2]{\pi_{#1} \ #2}
\newcommand{\join}[2]{#1 \bowtie #2}
\newcommand{\union}[2]{#1 \cup #2}
\newcommand{\fold}[3]{\t{fold} \ #1 \ #2 \ #3}

\newcommand{\gray}[1]{\textcolor{gray}{#1}}

\begin{figure}
\centering
\[
\begin{array}{llr}
        \mydefhead{e ::=\qquad\qquad\qquad\qquad}{Term}
        \mydefcase{\gray{x}}{\gray{variable}}
        \mydefcase{\gray{c}}{\gray{constant}}
        \mydefcase{\gray{\lam x e}}{\gray{abstraction}}
        \mydefcase{\gray{\ue_1~\ue_2}}{\gray{application}}
        \mydefcase{\gray{\mkref e}}{\gray{reference allocation}}
        \mydefcase{\gray{\deref e}}{\gray{dereference}}
        \mydefcase{\gray{\assign {e_1} {e_2}}}{\gray{assignment}}
        \mydefcase{\gray{\labexpr{k}{e_H}{e_L}}}{\gray{faceted expression}}
        \mydefcase{\gray{\levelDecl{k}{e}}}{\gray{label declaration}}
        \mydefcase{\gray{\policy{k}{e}}}{\gray{policy specification}}
        \mydefcase{\row{\overline{e}}}{create a table}
        \mydefcase{\select{i}{j}{e}}{select rows where $i=j$}
        \mydefcase{\project{\overline{i}}{e}}{project columns}
        \mydefcase{\join{e_1}{e_2}}{join or cross-product of tables}
        \mydefcase{\union{e_1}{e_2}}{union of tables}
        \mydefcase{\fold{e_f}{e_p}{e_t}}{table fold}
       \\
        \mydefhead{\gray{S ::=\qquad\qquad\qquad\qquad}}{\gray{Statement}}
        \mydefcase{\gray{\letin{x}{e}{S}}}{\gray{let statement}}
        \mydefcase{\gray{\print{e_v}{e_r}}}{\gray{print statement}}
        \\
        \mydefhead{c::=}{Constant}
        \mydefcase{\gray{\fh}}{\gray{file handle}}
        \mydefcase{\gray{b}}{\gray{boolean}}
        \mydefcase{\gray{i}}{\gray{integer}}
        \mydefcase{s}{string}
\\
        \mydefhead{\gray{x,y,z}}{\gray{Variable}}
        \mydefhead{\gray{k,l}}{\gray{Label}}
\end{array}
\]
\caption{\dblang{} syntax.}
\label{fig:core_syntax}
\end{figure}

The language \dblang extends the language \corelang~\cite{AYPLAS13} with support for databases, which we model as relational tables. Figure~\ref{fig:core_syntax} summarizes the \dblang syntax, with the constructs from \corelang marked in gray. The \corelang language, in turn, extends the standard imperative $\lambda$-calculus with constructs for declaring new labels ($\levelDecl{k}{e}$), for imperatively attaching policies to labels ($\policy{k}{e}$), and for creating faceted values ($\labval{k}{e_H}{e_L}$). This last expression behaves like $e_H$ from the perspective of any principal authorized to see data with label $k$ and $e_L$ for all other principals. Note that \dblang does not include imperative updates to tables, but we can model updates by introducing a layer of indirection where we access tables via references and updating a table corresponds to replacing the contents of the appropriate reference.

The language \dblang extends \corelang with support for databases, where each table is a (possibly empty) sequence of rows and each row is a sequence of strings. We require that all rows in a table have the same size. To manipulate tables, \dblang includes the usual operators of the relational calculus: \emph{selection} ($\select{i}{j}{e}$), which selects the rows in a table where fields $i$ and $j$ are identical, \emph{projection} ($\project{\overline{i}}{e}$), which returns a new table containing columns $\overline{i}$ from the table $e$, \emph{cross-product} ($\join{e_1}{e_2}$), which returns all possible combinations of rows from $e_1$ and $e_2$, and \emph{union} ($\union{e_1}{e_2}$), which appends two tables.
The construct $\row{\overline{e}}$ creates a new single-row table. The fold operation $\fold{e_f}{e_p}{e_t}$ supports iterating, or folding, over tables. Fold has the ``type'' $\forall \overline{A}, B. (B \rightarrow \overline{A} \rightarrow B) \rightarrow B \rightarrow \jtable{\overline{A}} \rightarrow B$.


\subsection{Formal Semantics}
\begin{figure*}
\input{standard_facet_rules}
\caption{Faceted evaluation of \dblang without relational operators.\label{fig:fvSem}\label{fig:runSyn}\label{fig:refUtils}}
\end{figure*}

\begin{figure*}
\vspace{-3ex}
\[
\begin{array}[t]{cl}
\setuprulecol
\mmrule{f-row}{
  }{
    \sstep \Sigma {\row{\overline{s}}} {pc} \Sigma {(\jtable{(\epsilon, \overline{s})})}
  }
\mmrule{f-select}{
    T' = \{ (B, s_1 \dots s_n) \in T ~|~ s_i = s_j \}
  }{
    \sstep \Sigma {\select{i}{j}{(\jtable{T})}} {pc} \Sigma {(\jtable{T'})}
  }
\end{array}
\begin{array}[t]{cl}
\setuprulecol
\mmrule{f-union}{
  }{
    \sstep \Sigma {\union{(\jtable{T_1})}{(\jtable{T_2})}} {pc} \Sigma {(\jtable{T_1.T_2})}
  }
\mmrule{f-project}{
    \overline{i} = i_1 \dots i_n \qquad
    T' = \{ (B, s_{i_1} \dots s_{i_n}) ~|~ (B, s_1 \dots s_m) \in T\}\\
  }{
    \sstep \Sigma {\project{\overline{i}}{(\jtable{T})}} {pc} \Sigma {(\jtable{T'})}
  }
\end{array}
 \]
\vspace{-3ex}
\[
\begin{array}[t]{cl}
\vspace{-2ex}
\mmrule{f-join}{
     T_3 = \{ (\union{B_1}{B_2}, s_1 \dots s_m s_1' \dots s_n') ~|~ (B_1, s_1 \dots s_m) \in T_1, (B_2, s_1' \dots s_n') \in T_2 \}
  }{
    \sstep \Sigma {\join{(\jtable{T_1})}{(\jtable{T_2})}} {pc} \Sigma {(\jtable{T_3})}
  }
\mmrule{f-fold-empty}{

  }{
    \sstep {\Sigma} {\fold{V_f}{V_p}{(\jtable{\epsilon})}} {pc} {\Sigma} {V_p}
  }
\mmrule{f-fold-inconsistent}{
  \sstep {\Sigma} {\fold{V_f}{V_p}{(\jtable{T})}} {pc} {\Sigma'} {V'} \qquad
  B \t{ inconsistent with } pc\\
}{
  \sstep {\Sigma} {\fold{V_f}{V_p}{(\jtable{(B, \jrow{s}).T})}} {pc} {\Sigma'} {V'}
}
\mmrule{f-fold-consistent}{
  \sstep {\Sigma} {\fold{V_f}{V_p}{(\jtable{T})}} {pc} {\Sigma'} {V'} \qquad
  B \t{ consistent with } pc \qquad
  \sstep {\Sigma'} {V_f~\jrow{s}~V'} {pc \cup B} {\Sigma''} {V''}
}{
  \sstep {\Sigma} {\fold{V_f}{V_p}{(\jtable{(B, \jrow{s}).T})}} {pc} {\Sigma''} {\labResult{B}{V''}{V'}}
}
\end{array}
\]
\vspace{-2ex}
\caption{Faceted evaluation with relational operators.}\label{fig:relOps}
\end{figure*}

We formalize the big-step semantics as the relation $\sstep {\Sigma} {e} {pc} {\Sigma'} {V}$, denoting that expression $e$ and store $\Sigma$ evaluate to $V$, producing a new store $\Sigma'$. The program counter $pc$ is a set of \emph{branches}. Each branch is either a label $k$ or a negated label $\notlab{k}$. Association with $k$ means the computation is visible only to principals authorized to see $k$ and association with $\notlab{k}$ visibility only to principals \emph{not} authorized to see $k$.

We chose our representation of faceted databases to be faithful to realistic implementation strategies. We could represent faceted tables as $\labexpr{k}{\jtable{T_1}}{\jtable{T_2}}$, but this approach would incur significant space overhead, as it requires storing two copies of possibly large database tables, possibly with only small differences between the two tables. Instead, we use the more efficient approach of \emph{faceted rows}, where each row $(B, \overline{s})$ in the database includes a set of branches $B$ describing who can see that row. For example, the expression $\labexpr{k}{\row{\sql{"Alice"}~{\sql{"Smith"}}}}{\row{\sql{"Bob"}~\sql{"Jones"}}}$ evaluates to the following table~\footnote{Note that this value representation does not support mixed expressions such as $\labexpr{k}{3}{\row{\mbox{"Alice"}}}$, which mix integers and tables in the same faceted values. Programs that try to unnaturally mix values will get stuck.}:
\begin{align*}
(\{k\}, (\sql{"Alice"}, \sql{"Smith"}))\\
(\{\notlab{k}\}, (\sql{"Bob"}, \sql{"Jones"}))
\end{align*}
Note that we do not model the facet identifier row \python{jid}, as it is not necessary for the formal semantics or proof.

To accommodate both faceted values and faceted tables, we define the partial operation $\labResult \cdot \cdot \cdot$ to create either a new faceted value or a table with internal branches on rows:
\[
\begin{array}{lcl}
  \labResult \cdot \cdot \cdot &:& \ULabel \times \UValue \times \UValue \rightarrow \UValue\\
  \labResult{k}{F_H}{F_L} &\defeq& \labexpr{k}{F_H}{F_L}\\
  \labResult{k}{\jtable{T_H}}{\jtable{T_L}} &\defeq& \jtable{T}\\
  \multicolumn{3}{l}{\mbox{where}~T = \{ (B, \overline{s}) ~|~ (B, \overline{s}) \in T_H \cap T_L \} \cup}\\
  \multicolumn{3}{l}{\quad\quad\quad\quad\quad\quad \{ (B \cup \{ k \}, \overline{s}) ~|~ (B, \overline{s}) \in T_H \setminus T_L, \notlab{k} \not\in B \} \cup}\\
  \multicolumn{3}{l}{\quad\quad\quad\quad\quad\quad \{ (B \cup \{ \notlab{k} \}, \overline{s}) | (B, \overline{s}) \in T_L \setminus T_H, k \not\in B \}}\\
\end{array}
\]
Wrapping a facet with label $k$ around non-table values $F_H$ and $F_L$ simply creates a faceted value containing $k$, $F_H$, and $F_L$. Wrapping a facet with label $k$ around tables $T_H$ and $T_L$ creates a new table $T$ containing the rows from $T_H$ and $T_L$, annotated with $k$ and $\notlab{k}$ respectively, with an optimization to share the rows that $T_H$ and $T_L$ have in common.
We extend this operator to sets of branches:
\[
\begin{array}{@{}lcl}
  \labResult{\cdot}{\cdot}{\cdot} &:& \USignedLabels \times \UValue \times \UValue \rightarrow \UValue\\
  \labResult{\emptyset}  {V_H} {V_L}
    &\defeq&  V_H \\
  \labResult {\{ k \} \cup B}{V_H}{V_L}
    &\defeq&   \labResult k {\labResult {B} {V_H} {V_L}} {V_L} \\
  \labResult {\{ {\notlab k} \} \cup B}{V_H}{V_L}
    &\defeq&   \labResult k {V_L} {\labResult B {V_H} {V_L}} \\
\end{array}
\]

We show the faceted evaluation rules in Figures~\ref{fig:runSyn} and~\ref{fig:relOps}. The key rule is \rel{f-split}, describing how evaluation of a faceted expression $\labexpr{k}{e_1}{e_2}$ involves evaluating the sub-expressions in sequence.
Evaluation adds $k$ to the program counter to evaluate $e_1$ and $\notlab k$ to evaluate $e_2$ and then joins the results in the operation $\labResult k {V_1}{V_2}$. The rules \rel{f-left} and \rel{f-right} show that only one expression is evaluated if the program counter already contains either $k$ or $\notlab k$.

Our rules use contexts to describe faceted execution. The rule \rel{f-ctxt} for $E[e]$ enables evaluation of a subexpression inside an evaluation context. We use $S$ to range over strict operator contexts, operations that require a non-faceted value. If an expression in a strict context yields a faceted value $\labexpr{k}{V_H}{V_L}$, then the rule \rel{f-strict} applies the strict operator to each of $V_H$ and $V_L$. For example, the evaluation of $\labexpr{k}{f}{g}(4)$ reduces to the evaluation of $\labexpr{k}{f(4)}{g(4)}$, where $S$ in this case is $\ctxthole(4)$. The rules \rel{f-select}, \rel{f-select}, \rel{f-proj}, \rel{f-join}, and \rel{f-union} formalize the relational calculus operators on tables of faceted rows.

The rules for folding are more interesting. If a row $(B, \overline{s})$ is inconsistent with (\ie, not visible to) the current program counter label $pc$, then rule \rel{f-fold-inconsistent} ignores that row. If the row is consistent, then rule \rel{f-fold-consistent} applies the fold operator $V_f$ to the row contents $\overline{s}$ and the accumulator $V'$, producing a new accumulator $V''$. The result of that fold step is $\labResult{B}{V''}{V'}$, a faceted expression that appears like $V''$ to principals that can see the $B$-labeled row and like $V'$ to other principals.

The faceted execution semantics describe the propagation of labels and facets for the purpose of complying with policies at computation sinks. \dblang expressions do not perform I/O, while \dblang statements include the effectful construct $\print{e_v}{e_r}$ that prints expression $e_r$ under the policies and viewing context $e_v$. We provide the \corelang rules for declaring labels, attaching policies, and assigning labels for printing in Appendix~\ref{sec:extra_rules}. The 
\python{@label_for} and \python{jacqueline_restrict} constructs correspond to thhe \rel{f-label} and \rel{f-restrict} rules.


\newcommand{\lemProjLabel}{
\begin{align*}
L(\labResult{k}{V_1}{V_2}) &= \left\{ \begin{array}{lr}
                              L(V_1) & \t{if } k \in L\\
                              L(V_2) & \t{if } k \not\in L\\
                              \end{array} \right.\\
\end{align*}
}

\newcommand{\lemProjBranch}{
\begin{align*}
L(\labResult{B}{V_1}{V_2}) &= \left\{ \begin{array}{lr}
                              L(V_1) & \t{if } \visible{B}{L}\\
                              L(V_2) & \t{if } \neg(\visible{B}{L})\\
                              \end{array} \right.\\
\end{align*}
}

\newcommand{\lemInvis}{
If $pc$ is not visible to $L$ and
\begin{equation*}
\sstep {\Sigma} {e} {pc} {\Sigma'} {V}
\end{equation*}
then $L(\Sigma) = L(\Sigma')$.
}

\newcommand{\thmProjection}{
Suppose $\sstep {\Sigma} {e} {pc} {\Sigma'} V$.
Then for any view $L$ for which $pc$ is visible,
\begin{align*}
\sstep {L(\Sigma)} {L(e)} {\emptyset} {L(\Sigma')} {L(V)}
\end{align*}
}

\subsection{Application-Database Policy Compliance}
\corelang~\cite{AYPLAS13} has the properties that 1) a single faceted execution is equivalent to multiple different executions without faceted values and 2) the system cannot leak sensitive information through the output or the choice of output channel. We prove that the properties extend to \dblang.

The proof involves extending the \emph{projection} property of \corelang: a single execution with faceted values \emph{projects} to multiple different executions without faceted values. To prove this property, we first define what it means to be a \emph{view} and to be \emph{visible}.
A \emph{view} $L$ is a set of principals. $B$ is visible to view $L$ (written $\visible{B}{L}$) if $\forall k \in B . k \in L$ and $\forall \notlab{k} \in B . k \not\in L$.
We extend views to values:
\begin{align*}
L: Val \t{(with facets)} &\rightarrow Val \t{(without facets)}\\
L(R) &= R\\
L(\labval{k}{F_1}{F_2}) &= \left\{ \begin{array}{lr}
                            L(F_1) & k \in L\\
                            L(F_2) & k \not\in L\\
                            \end{array} \right.\\
L(\jtable{T}) &= \{ (\emptyset, \jrow{s}) ~|~ (B, \jrow{s}) \in T, B \t{ visible to } L\}\\
\end{align*}
We extend views to expressions:
\begin{align*}
L(\labval{k}{e_1}{e_2}) &= \left\{ \begin{array}{lr}
                            L(e_1) & k \in L\\
                            L(e_2) & k \not\in L\\
                            \end{array} \right.\\
\end{align*}
For all other expression types we recursively apply the view to subexpressions.

We then prove the Projection Theorem. The full proof is in Appendix~\ref{sec:projection}. Proofs of the key lemmas are in Appendices~\ref{sec:lemma_projlabel} and~\ref{sec:lemma_projbranch}.
\begin{thm}[Projection]
\label{thm:projection}
\thmProjection
\end{thm}
\noindent The Projection Theorem allows us to extend \corelang's property of termination-insensitive non-interference. To state the theorem we first define two faceted values to be \emph{$L$-equivalent} if they have identical values for the view $L$. This notion of $L$-equivalence naturally extends to stores $(\Sigma_1 \equivlab{\pc} \Sigma_2)$ and expressions $(e_1 \equivlab{\pc} e_2)$. The theorem is as follows:

\begin{thm}[Termination-Insensitive Non-Interference]
\label{thm:noninterference} \mbox{}\\
Let $L$ be any view. Suppose $\Sigma_1\equivlab{L} \Sigma_2$ and $e_1 \equivlab{L} e_2$, and that:
\[
\begin{array}{l@{\qquad}l}
  \sstep{\Sigma_1} {e_1} {\emptyset} {\Sigma_1'} {V_1} &
  \sstep{\Sigma_2} {e_2} {\emptyset} {\Sigma_2'} {V_2} \\
\end{array}
\]
~~~~then
$\Sigma_1' \equivlab{L}\Sigma_2'$ and
$V_1  \equivlab{L} V_2$.
\end{thm}
\noindent The Termination-Insensitive Non-Interference Theorem allows us to extend the termination-insensitive \emph{policy compliance} theorem of \corelang~\cite{AYPLAS13}: data is revealed to an external observer only if it is allowed by the policies specified.

\subsection{Early Pruning}
\label{sec:early_pruning_rule}
The Early Pruning optimization involves shrinking a table $T$ by keeping each row $(B, \overline{s})$ only when $B$ is consistent with the viewer constraint described by $pc$. We show the rule below:
\[
\begin{array}{rclcl}
\mmrule{f-prune}{
    \sstep {\Sigma} {e} {pc} {\Sigma'} {(\jtable{T})}\\
    T' = \{ (B, \jrow{s}) \in T ~|~ B \t{ consistent with } {pc} \}\\
  }{
    \sstep {\Sigma} {e} {pc} {\Sigma'} {(\jtable{T'})}
  }
\end{array}
\]
We prove the Projection Theorem holds with this extension.

\subsection{Policy Dependencies on Sensitive Values}
Policies on a label may contain sensitive values that depend on the label. The \dblang semantics handles this situation by propagating labels through all computations and then assigning label values according to the \rel{f-print} rule from \corelang (Appendix~\ref{sec:extra_rules}). Our theorem modifies the traditional notion of non-interference to accommodate these dependencies. In the statement of Theorem~\ref{thm:noninterference}, the resulting environments need only be $L$-equivalent when the viewer \emph{does not} have access to the high-confidentiality view. If a viewer does not have access, then the sensitive value should be indistinguishable from any other value the viewer does not have access to.

\section{Implementation}
\label{sec:implementation}
While previous implementations of Jeeves~\cite{JeevesPaper, AYPLAS13} use Scala, we implement \webframework in Python, as an extension of Django~\cite{Django}, because of the popularity of both for web programming.
Our code is available at \url{https://github.com/jeanqasaur/jeeves}.

\subsection{Python Embedding of the Jeeves Runtime}
We implemented Jeeves as a library that dynamically rewrites code to behave according to the \corelang semantics.
The library exports functions for creating labels, creating sensitive values, attaching policies, and using policies to show values. Our implementation supports a subset of Python's syntax that includes if-statements, for-loops, and return statements.





\subsubsection{Faceted Execution}
To support faceted execution, the implementation defines a \python{Facet} data type for primitives and objects where the facets may themselves be faceted. A value may exist only in some execution paths, in which case we use a special object \python{Unassigned()} for other paths. To perform faceted execution, the implementation uses operator overloading and dynamic source transformation via the macro library \python{MacroPy}~\cite{MacroPy}. The source transformation intercepts evaluation of conditionals, loops, assignments, and function calls.
The implementation handles local assignment by replacing a function's local scope with a \python{Namespace} object determining scope. To prevent implicit flows, the runtime keeps track of path conditions to index state updates, database writes, and policy declarations.
\subsubsection{Evaluating Policies at Computation Sinks}
The runtime maps labels to policies. If there are no mutual dependencies between policies and sensitive values, the runtime evaluates policies to determine label values. Otherwise, the runtime produces an ordering over Boolean label assignments and uses the SAT subset of the Z3 SMT solver~\cite{Z3} to find a satisfying assignment.

\subsection{\webframework Implementation}
We extend Django's functionality by ``monkey-patching,'' inheriting from Django's classes and overloading the methods of the FORM. The FORM is responsible for 1) marshalling between faceted representations in the application and database and 2) managing the meta-data to track facets in the database. To represent faceted values, the FORM creates schemas with additional meta-data columns. The FORM reconstructs facets from the meta-data by looking up policies from object schemas and adding them to the runtime environment. We implement the Early Pruning optimization by reconstructing only the relevant facets when the runtime knows the viewer. FORM queries manipulate the meta-data columns in addition to the actual columns. Programmers may access the database only through the supported API.

\section{\webframework in Practice}
\label{sec:evaluation}
We built 1) a conference management system, 2) a health record manager, and 3) a course management system to evaluate \webframework along the following dimensions:
\begin{itemize}
\item \textbf{Code architecture.} We compare the implementation of the \webframework conference management system to an implementation with hand-coded policies in Django. We demonstrate that \webframework helps with both centralizing policies and with size of policy code.
\item \textbf{Performance.} We show that for representative actions, \webframework has comparable---and, in one case, better---performance compared to Django. For the stress tests, the \webframework programs often have close to zero overhead and at most a 1.75x slowdown compared to vanilla Django. We also demonstrate the effectiveness of and necessity of the Early Pruning optimization.
\end{itemize}
While the conference management system has the most features, the applications have code and policies of similar complexity. We deployed our conference management system to run an academic workshop~\cite{PLOOC2014}.

We worked with two undergraduate research assistants to implement the health record and course manager case studies to evaluate the usability of our programming model. Both students found the policy-agnostic approach to be ``promising'' and to hide complexity in implementing information flow policies. There were some objections about the boilerplate needed to use the Jeeves and \webframework libraries, as well as the acknowledgment that building these features into a language runtime directly would mitigate this issue.

\subsection{Case Study Applications}
\noindent \textbf{Conference management system.} We support user registration, update of profile information, designation of roles (\ie{} PC member), paper and review submission, and assignment of reviews.
Permissions depend on the current stage of the conference: submission, review, or decision.

\noindent \textbf{Health record manager.} We implemented a simple health record system based on a representative fragment of the privacy standards described in the Health Insurance Portability and Accountability Act (HIPAA)~\cite{HIPAASummary, BarthHIPAA}. HIPAA describes how individuals, hospitals, and insurance companies may view a medical history depending on roles and stateful information such as whether there exists a permission waiver. The case study manages health records and permissions when viewed by patients, doctors, and insurance companies.

\noindent \textbf{Course manager.} Our tool allows instructors and students to organize assignments and submissions. Policies depend on the role of the viewer, as well as stateful information such as whether an assignment has been submitted.

\subsection{Code Comparisons}
We compare our \webframework implementation of a conference management system against a Django implementation of the same system. We demonstrate that 1) \webframework reduces the trusted computing base and 2) separating policies and other functionality decreases policy code size. 

\subsubsection{Django Conference Management System}
\begin{figure}
\centering
\textbf{Lines of Policy Code: \webframework vs. Django}\\
\includegraphics{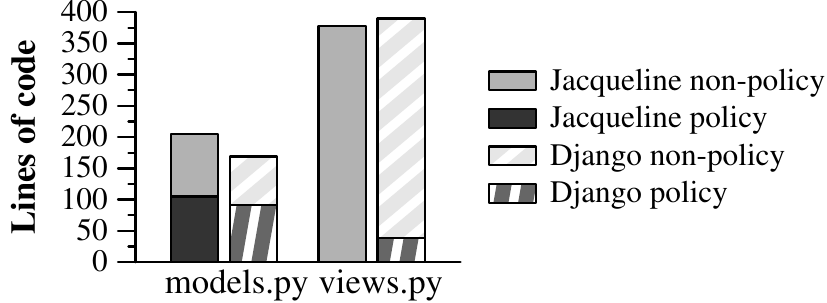}
\caption{Distribution of policy code with \webframework and Django conference management systems.}
\label{fig:django_loc_comparison}
\end{figure}

We compare the lines of code in the \webframework and Django conference management systems in Figure~\ref{fig:django_loc_comparison}. (Note that \webframework counts are bloated from the additional imports and function decorators required.)
\webframework demonstrates advantages in both the distribution and size of policy code. In the \webframework implementation, policy code is confined to the \python{models.py} file describing the data schemas, while in the Django implementation, there are also policies throughout the controller file \python{views.py}. The \webframework implementation has 106 total lines of policy code, whereas the Django implementation has 130 lines manifesting as repeated checks and filters across \python{views.py}. While the Django code requires auditing the $575$ lines of \python{models.py} and \python{views.py}, the \webframework code requires auditing only the $200$ lines of \python{models.py} (\textasciitilde$200$ lines of code), reducing the size of the application-specific trusted code base by 65\%.

\begin{figure}
\begin{lstlisting}[language=Python, xleftmargin=10pt, numbers=left, escapeinside={~}{~}]
class Paper(Model):
  ... 
  @staticmethod
  @label_for('author')
  @jeeves
  def jeeves_restrict_author(paper, ctxt):
    if phase == 'final':
      return True
    else:
      if paper == None:
        return False
      if PaperPCConflict.objects.get(
          paper=paper, pc=ctxt) != None:
        return False
      return ((paper != None and
        paper.author == ctxt)
        or (ctxt != None and
              (ctxt.level == 'chair' or
                ctxt.level == 'pc')))
\end{lstlisting}
\caption{Jacqueline schema and code fragments.}
\label{fig:jacqueline_cfm_code}
\end{figure}

\begin{figure}
\textbf{Django Schema}
\begin{lstlisting}[language=Python, xleftmargin=10pt, numbers=left, escapeinside={~}{~}]
class Paper(Model):
  ...
  def policy_author(self, ctxt):
    if phase == 'final':
      return True
    else:
      try:
        conflict =
          PaperPCConflict.objects.get(
            paper=self, pc=ctxt)
        return False
      except:
        return ((self.author == ctxt)
          or (ctxt != None and
            (ctxt.level == 'chair' or
              ctxt.level == 'pc')))
\end{lstlisting}

\textbf{Python Code with Policy Checks}
\begin{lstlisting}[language=Python, xleftmargin=10pt, numbers=left, escapeinside={~}{~}]
def papers_view(request):
  papers = Paper.objects.all()
  for paper in papers:
    if not paper.policy_paperlabel(user):
      paper.author = None
\end{lstlisting}
\caption{Django schema and code fragments.}
\label{fig:django_cfm_code}
\end{figure}

We show a fragment of \webframework policy code in Figure~\ref{fig:jacqueline_cfm_code} and a fragment of the analogous Django policy code in Figure~\ref{fig:django_cfm_code}, along with an example of how the policy check functions are called in the Django implementation. The policy functions are similar across the two implementations, with small differences from the fact that the \webframework API to the database does not raise an exception if an entry is missing. The main difference is that in the Django implementation, the programmer is responsible for calling these policy functions (as we show in Figure~\ref{fig:django_cfm_code}), whereas in the \webframework implementation the runtime is responsible for handling the interaction with policy functions.

\subsection{Performance Measurements}
We measured times using an Amazon EC2 m3.2xlarge instance running Ubuntu 14.04 with 30GB of memory, two 80GB SSD drives, and eight virtual 64-bit Intel(R) Xeon(R) CPU E5-2670 v2 2.50Ghz processors. We use the FunkLoad testing framework~\cite{FunkLoad} for HTTP requests across the network, excluding CSS and images. We average over 10 rapid sequential requests. We test with sequential users because how well \webframework handles concurrent users compared to Django simply depends on the amount of available memory.

We show 1) policy enforcement in \webframework has reasonable overheads, especially compared to Django and 2) Early Pruning is effective and often necessary.

\subsubsection{Stress Tests}
\begin{figure}
\centering
\includegraphics{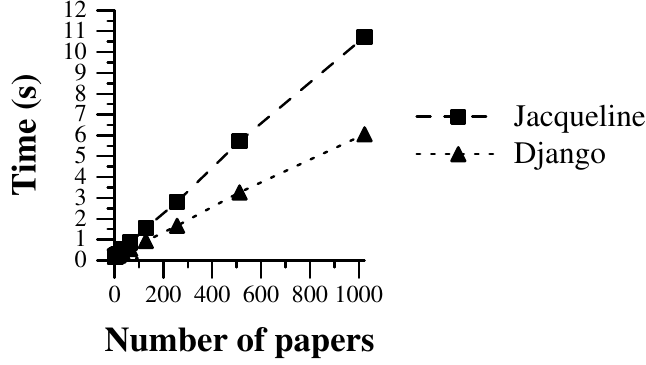}
\includegraphics{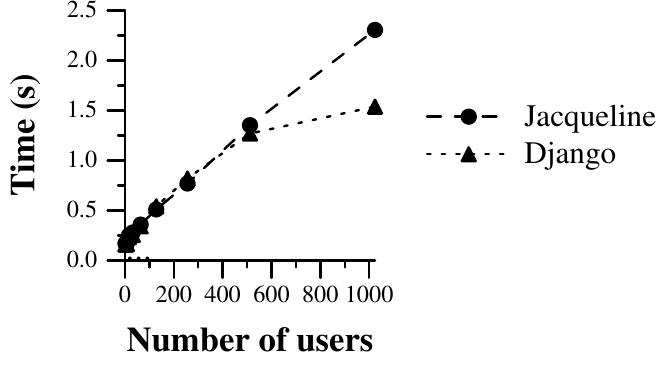}

(a) Conference management system

\begin{tabular}{cc}
\includegraphics{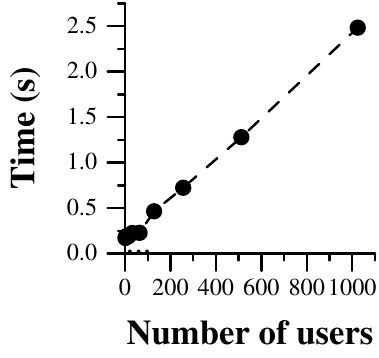} & \includegraphics{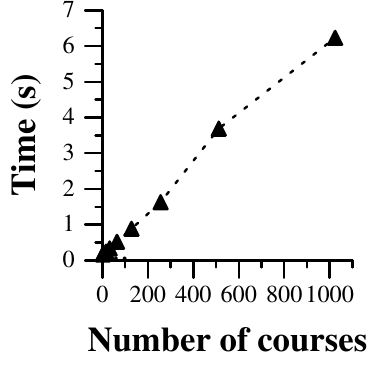}\\
(b) Health record manager & (c) Course manager\\
\end{tabular}

\caption{Stress test times for our three case studies. We compare the conference management system times (a) to an implementation in Django.}
\label{fig:stress_tests}
\end{figure}

\begin{table}
\begin{tabular}{|r|r|r|}
\hline
\multicolumn{3}{|c|}{Time to view all papers}\\
\hline
\# P & Jacq. & Django\\
\hline
\hline
8 & 0.241s & 0.201s\\
16 & 0.299s & 0.241s\\
32 & 0.542s & 0.388s\\
64 & 0.855s & 0.554s\\
128 & 1.551s & 0.931s\\
256 & 2.810s & 1.633s\\
512 & 5.717s & 3.265s\\
1024 & 10.729s & 6.055s\\
\hline
\end{tabular}
\begin{tabular}{|r|r|r|}
\hline
\multicolumn{3}{|c|}{Time to view all users}\\
\hline
\# U & Jacq. & Django\\
\hline
\hline
8 & 0.172s & 0.163s\\
16 & 0.249s & 0.234s\\
32 & 0.279s & 0.254s\\
64 & 0.358s & 0.341s\\
128 & 0.510s & 0.541s\\
256 & 0.769s & 0.820s\\
512 & 1.352s & 1.269s\\
1024 & 2.305s & 1.538s\\
\hline
\end{tabular}
\label{tab:cfm_stress_tests}
\caption{Times to view a list of summary information for conference manager stress tests, in \webframework and Django.}
\end{table}

In Figure~\ref{fig:stress_tests} we show running times from our stress tests. For each application, we show an increasing number of a given type of data item. The graphs demonstrate that with both \webframework and Django, the time to load data scales linearly with respect to the underlying algorithms.  
The numbers (Table~\ref{tab:cfm_stress_tests}) show that \webframework has at most a 1.75x overhead. The overhead comes from fetching both versions of data before resolving the policies. There is no solver overhead, as there are no mutual dependencies between sensitive values and policies.
Note that these are truly stress tests: most systems will not load a thousand data rows at once, especially when each value has its own policy involving database queries.


\subsubsection{Representative Actions}
\begin{table}
\centering
\begin{tabular}{|r|r|r|}
\hline
\multicolumn{3}{|c|}{Time to view single paper}\\
\hline
Papers & Jacq. & Django\\
\hline
\hline
8 & 0.160s & 0.177s\\
16 & 0.165s & 0.175s\\
32 & 0.160s & 0.177s\\
64 & 0.159s & 0.173s\\
128 & 0.160s & 0.173s\\
256 & 0.159s & 0.173s\\
512 & 0.159s & 0.178s\\
1024 & 0.161s & 0.173s\\
\hline
\end{tabular}
\begin{tabular}{|r|r|r|}
\hline
\multicolumn{3}{|c|}{Time to view single user}\\
\hline
Users & Jacq. & Django\\
\hline
\hline
8 & 0.164s & 0.158s\\
16 & 0.164s & 0.159s\\
32 & 0.164s & 0.159s\\
64 & 0.164s & 0.159s\\
128 & 0.167s & 0.158s\\
256 & 0.163s & 0.159s\\
512 & 0.169s & 0.162s\\
1024 & 0.163s & 0.159s\\
\hline
\end{tabular}
\caption{Times to view profiles for a single paper and single user, in \webframework and Django.}
\label{tab:representative_actions}
\end{table}

We increased the number of relevant database entries and measured the time it takes to view the profiles for single papers and users. We show these numbers, as well as comparisons to Django, in Table~\ref{tab:representative_actions}. The time it takes to load these profiles is under $2$ms and roughly equivalent to the time it takes to do the equivalent action in Django. For viewing a single paper, \webframework actually performs better than the Django implementation. This is because in the Django code, the implementation needs iterate over collections of data rows again in order to apply policy checks. In the \webframework implementation, the framework applies the policies and resolves each one once. Times for submitting a single paper scale similarly.

\subsubsection{Early Pruning Optimization}
\begin{table}
\centering
\begin{tabular}{|r|r|r|}
\hline
Courses & Time w/o pruning & Time w/ pruning\\
\hline
\hline
4 & 0.377s & 0.185s\\
8 & 64.024s & 0.192s\\
16 & -- & 0.248s\\
32 & -- & 0.337s\\
64 & -- & 0.522s\\
128 & -- & 0.886s\\
256 & -- & 1.630s\\
512 & -- & 3.691s\\
1024 & -- & 6.233s\\
\hline
\end{tabular}
\caption{Showing all courses, with and without Early Pruning.}
\label{tab:early_pruning}
\end{table}

We found the Early Pruning optimization to be necessary for nontrivial computations over sensitive values. In the course manager stress test, the page that shows all courses also looks up the instructors for each course, leading to blowup. We show in Table~\ref{tab:early_pruning} how for just eight courses and instructors, the system begins to hit memory limits.
Because Early Pruning can simplify other computations after the viewer is known, these computations are only problematic when they are used to compute the viewer. We do not expect such computations to be common.

\section{Limitations and Future Work}
\label{sec:discussion}

The policy-agnostic approach does not protect against a malicious programmer who implements incorrect policies. By centralizing the policies and making the implementations more concise, however, we hope to make it easier to audit programs to determine policies are implemented correctly.

The strategy of embedding the programming model in a Python web framework requires the programmer to enforce certain invariants. Because sensitive values exist unprotected in the program runtime and database, the programmer must access sensitive values only through the designated APIs. Implementing Jeeves in a language with private class attributes would alleviate some of these concerns.

Our strategy of using two database rows to encode each faceted value means that we cannot use existing database support for performing aggregates or optimizing based on primary keys. A solution that would also reduce database size is to compute low-confidentiality values upon retrieving data from the database, rather than storing the values in the database. In cases when this is not possible, an alternate strategy is to implement user-defined functions in the database to optimize based on the Jeeves-based unique ID for each faceted value, as well as for aggregates.

Another future direction involves optimizing queries to reduce the amount of data fetched based on the policies associated with the data.

It would also be useful to extend policy-agnostic programming with faceted values to operating systems. We can build on the techniques from the Laminar system~\cite{Roy2009}, which demonstrates how to dynamically enforce policies mediating access to resources such as files and sockets.

\section{Related Work}
Our approach builds on a long history of work in information flow control~\cite{DenningDenning77, Jif0, FabricSOSP, OaklandFabric, FineCompiler, FStar, FlowLocks, Chugh2009, pottier-simonet-toplas-03, Krohn2007, Resin, Roy2009}. The policy-agnostic approach differs from prior work in the following key way. Using prior approaches, the programmer needs to implement the policy checks and filters correctly across the program. Our solution mitigates programmer burden by leveraging the language runtime to produce outputs adhering to policies. This is similar in philosophy to angelic nondeterminism~\cite{Bodik2010}, program repair~\cite{PlanB, Samimi2012}, and acceptability-oriented computing~\cite{Rinard2004, Rinard2005}.

Prior work on information flow across the application-database boundary focuses on rejecting queries that leak information, rather than on modifying queries to enforce policies. SeLINQ~\cite{SeLINQ}, the work of Louren{\c{c}}o and Caires~\cite{IFDataSecurity}, and Ur/Web use static types. DBTaint~\cite{DavisC10}, Passe~\cite{passe14}, and Hails~\cite{Hails} perform dynamic analysis. SIF~\cite{SIF} combines static labels and dynamic checks. There are also approaches based on symbolic execution~\cite{rozzleTR}, secure multi-execution~\cite{shadowexecution, SecureMultiexecution, WebMultiExecution}, and analysis of data provenance~\cite{ProvenanceSecurity, ProvenanceSecurityConfiguration} focused on rejecting programs that violate desired properties.

Policy-agnostic programming differs from other approaches in how data may affect control flow. Variational data structures~\cite{Walkingshaw} encapsulate properties related to program customization, but data does not affect control flow. Aspect-oriented programming~\cite{AOP, GenerativeAOP} has similar goals to policy-agnostic programming of separating program concerns, but aspects must be implemented at specific control flow points and cannot alter control flow.

Our approach addresses information flow as opposed to access control~\cite{Rubicon, Margave, Sunny}, which prevents leaks at application endpoints and does not address indirect or implicit flows. Similarly, work on multi-level databases~\cite{DenningAMNSH86, Lunt1990} focuses on the storage and access control issues surrounding data at different levels of access in the database.

\section{Conclusions}
We demonstrate that it is practically feasible to achieve policy compliance by construction in database-backed applications. We present a technique for precise, dynamic information flow control that tracks sensitive values and policies through database queries and updates as well as application code. The technique supports a policy-agnostic programming model that allows the program to specify each information flow policy once, instead of as repeated intertwined checks across the program. The web framework performs different computations depending on the viewer, according to the policies. The shift of responsibility to the framework reduces the opportunity for programmer error to cause information leaks.

Our solution works with existing implementation of relational databases and yields formal guarantees across the application and database. We implement these ideas in the \webframework web framework and demonstrate that, compared to traditional applications with hand-coded policies, applications written using \webframework have less policy code and run with often negligible overheads. This work makes a promising step towards securing database-backed web applications.

\acks
We thank Martin Rinard, James Cheney, our shepherd Kathryn McKinley, and our anonymous reviewers for draft feedback. We thank Benjamin Shaibu and Ariel Jacobs for their work on the case studies. This work was supported by the Facebook and Levine Fellowships, the Qatar Computing Research Institute, and NSF grants 1054172, CCF-1139056, CCF-1337278, and CCF-1421016.
The views and conclusions contained herein are those of the authors and should not be interpreted as representing the official policies, either expressed or implied, of the U.S. Government.

\bibliographystyle{abbrvnat}
\bibliography{jacqueline}

\newcommand{\unionk}[2]{\union{#1}{\{#2\}}}
\newcommand{\cupk}{^{~\bigcup}_{k\in K}}

\appendix

\section{Rules from \corelang}
\label{sec:extra_rules}

We show the most relevant rules from the dynamic semantics for the Jeeves core language \corelang~\cite{AYPLAS13}.

\subsection{Managing Labels}
These rules describe how to declare labels and attach policies to labels.
The rule \rel{f-label} dynamically allocates a label ($\levelDecl{k}{e}$),
adding a fresh label to the store with the default policy of $\lam x \true$.
Any occurrences of $k$ in $e$ are $\alpha$-renamed to $k'$ and the expression
is evaluated with the updated store.
Policies may be further refined ($\policy k e$) by the rule \rel{f-restrict},
which evaluates $e$ to a policy $V$
that should be either a lambda or a faceted value comprised of lambdas.
The additional policy check is restricted by $\pc$,
so that policy checks cannot themselves leak data.
The rule joins the resulting policy check $V_p$ with the existing policy for $k$,
ensuring that policies can only become more restrictive.
\[
\begin{array}[t]{cl}
\mrule{f-label}{
                k' fresh \\
          \sstep {\Sigma[k':=\lam x \true]} {e[k:=k']} \pc {\Sigma'} V \\
  }{
    \sstep \Sigma {\levelDecl k e} \pc {\Sigma'} V'
  }
\mrule{f-restrict}{
    \sstep \Sigma e \pc {\Sigma_1} V \\
    V_p = \labResult{\joinLab{k}{\pc}}{V}{\lam x \true}\\
    \Sigma' = \Sigma_1[k:=\Sigma_1(k)\policyAdd V_p] \\
  }{
  \sstep \Sigma {\policy k e} {\pc} {\Sigma'} {V}
  }
\end{array}
\]

\subsection{Displaying Outputs}
The rule \rel{f-print} handles print statements ($\print{e_1}{e_2}$),
where the result of evaluating $e_2$ is printed to the channel resulting from the evaluation of $e_1$.
Both the channel $V_f$ and the value to print $V_c$ may be faceted values. The describes how to select the facets that correspond with our specified policies.
The rule determines the set of relevant labels through the transitive closure function $\fun{closeK}$. The labels are used to construct $e_p$ from the relevant policies in the store $\Sigma_2$.
The rule evaluates $e_p$ and applies it to $V_f$,
returning the policy check $V_p$ that is a faceted value containing booleans.
The rule chooses a program counter $\pc$ such that the policies are satisfied.
This corresponds to a label assignment that determines the channel $\fh$ and the value to print $R$.


\[
\begin{array}[t]{cl}
\mrule{f-print}{
  \sstep{\Sigma}{e_1}{\emptyset}{\Sigma_1}{V_f} \\
  \sstep{\Sigma_1}{e_2}{\emptyset}{\Sigma_2}{V_c} \\
  \{~ k_1 ~...~ k_n ~\} = \fun{closeK}(\fun{labels}(e_1) ~\cup~ \fun{labels}(e_2), \Sigma_2) \\
  e_p = \lam x \true ~\policyAdd~ \Sigma_2(k_1) ~\policyAdd~ ... ~\policyAdd~ \Sigma_2(k_n) \\
  \sstep{\Sigma_2}{e_p~V_f}{\emptyset}{\Sigma_3}{V_p} \\
  \mbox{pick}~\pc~\mbox{such that}~\pc(V_f)=f,
    \pc(V_c)=R,~ \pc(V_p)=\true
}{
  \stmtstep {\Sigma} {\print{e_1}{e_2}} {V_p} {\fileOut{f}{R}}
}
\end{array}
\]

\[
  \begin{array}{rcl}
    \fun{closeK}(K, \Sigma) & = &
           \mbox{let}~ K' = \cupk~ \fun{labels}(\Sigma(k)) ~\mbox{in}~ \\
        && \mbox{if}~ K' = K \\
        && ~~~~\mbox{then}~ K \\
        && ~~~~\mbox{else}~ \fun{closeK}(K', \Sigma) \\
  \end{array}
\]

\section{Proof of Lemma~\ref{lem:projLabel}}
\label{sec:lemma_projlabel}
\begin{lem}[A]
\label{lem:projLabel}
\lemProjLabel
\end{lem}
\begin{proof}
  By case analysis on the definition of $\labResult{k}{V_1}{V_2}$. \\
  Let $x = L(\labResult{k}{V_1}{V_2})$.
\begin{itemize}
  \item If $x = L(\labexpr{k}{F_1}{F_2})$ for some non-table values $F_1$ and $F_2$,
    then this case holds since
    \begin{itemize}
      \item $x = L(F_1)$ if $k\in L$.
      \item $x = L(F_2)$ if $k\not\in L$.
    \end{itemize}
  \item If $x = L(\labResult{k}{\jtable{T_1}}{\jtable{T_2}})$,
    then $x = L(\jtable {T})$
    where \\
    $T = \{ (\unionk{B}{k}, \overline{s}) ~|~ (B, \overline{s}) \in T_1, \notlab k\not\in B \}$ \\
    $~~~~\cup~ \{ (\unionk{B}{\notlab k}, \overline{s}) ~|~ (B, \overline{s}) \in T_2, k\not\in B \}$. \\
    And so \\
    $x = \{ (\emptyset,\overline{s}) ~|~ (B, \overline{s}) \in T_1, \notlab k\not\in B, \visible{\unionk{B}{k}}{L} \}$ \\
    $~~~~\cup~ \{ (\emptyset,\overline{s}) ~|~ (B, \overline{s}) \in T_2, k\not\in B, \visible{\unionk{B}{\notlab k}}{L} \}$.
    \begin{itemize}
      \item
        If $k\in L$,
        then $\nvisible{\unionk{B}{\notlab k}}{L}$ and \\
        $\visible{\unionk{B}{k}}{L} => \notlab k\not\in B$, and so \\
        $x = \{ (\emptyset, \overline{s}) ~|~ (B, \overline{s}) \in T_1, \visible{B}{L} \}$ \\
        $~~~= L(\jtable T_1)$, as required.
      \item
        If $k\not\in L$,
        then this case holds by a similar argument as the previous case.
    \end{itemize}
\end{itemize}
\end{proof}

\section{Proof of Lemma~\ref{lem:projBranch}}
\label{sec:lemma_projbranch}
\begin{lem}[B]
\label{lem:projBranch}
\lemProjBranch
\end{lem}
\begin{proof}
  The proof is by induction and case analysis on the derivation of $L(\labResult{B}{V_1}{V_2})$.
  Let $x = L(\labResult{B}{V_1}{V_2})$.
  \begin{itemize}
    \item
      If $B = \emptyset$, then $\visible{B}{L}$, so $x = L(V_1)$ as required.
    \item
      Otherwise, $B = \unionk{B'}{k}$.
      \begin{itemize}
        \item
          If $\visible B L$, then \\
          $x = L(\labResult{k}{\labResult{B'}{V_1}{V_2}}{V_2})$ \\
          $~~~= L(\labResult{B'}{V_1}{V_2})$ by Lemma~\ref{lem:projLabel}, since $k\in L$ \\
          $~~~= L(V_1)$ by induction, as $\visible{B'}{L}$.
        \item
          Otherwise, $\nvisible{B}{L}$, then
          \begin{itemize}
            \item if $k\not\in L$, then $x = L(V_2)$ by Lemma~\ref{lem:projLabel}.
            \item otherwise $k\in L$, so $\nvisible{B'}{L}$. \\
              Therefore, $x = L(\labResult{B'}{V_1}{V_2}) = L(V_2)$, as required.
          \end{itemize}
      \end{itemize}
  \end{itemize}
\end{proof}

\section{Lemma~\ref{lem:invis}}
If a set of branches is compatible with view $L$, then we can execute only using that view. We prove an additional lemma that if $pc$ is not visible, then execution should not affect the environment under projections of $L$.
\begin{lem}[C]
\label{lem:invis}
If $pc$ is not visible to $L$ and
\begin{equation*}
\sstep {\Sigma} {e} {pc} {\Sigma'} {V}
\end{equation*}
then $L(\Sigma) = L(\Sigma')$.
\end{lem}
\begin{proof}
  By induction on the derivation of
  $\sstep {\Sigma} {e} {\pc} {\Sigma'} {V}$
  and by case analysis on the final rule used in that derivation.
  \begin{itemize}
    \item The following cases hold because $\Sigma = \Sigma'$: \rel{f-val}, \rel{f-deref-null}, \rel{f-deref}, \rel{f-ctxt},
      \rel{f-app}, \rel{f-left}, \rel{f-right}, \rel{f-row}, \rel{f-select},
      \rel{f-project}, \rel{f-join}, \rel{f-union}, \rel{f-fold-empty}, and \rel{f-fold-inconsistent}.
    \item Cases \rel{f-app}, \rel{f-left}, \rel{f-right}, \rel{f-strict}, \rel{f-ctxt},
      and \rel{f-fold-inconsistent} hold by induction.
    \item For case \rel{f-split}, we note that since $\visible{\pc}{L}$,
      $\forall k.\nvisible{\unionk{\pc}{k}}{L}$ and $\nvisible{\unionk{\pc}{\notlab k}}{L}$.
      Therefore, this case also holds by induction.
    \item Similarly, for case \rel{f-fold-consistent}, since $\nvisible{\pc}{L}$, \\
      $\forall B.\nvisible{\union{\pc}{B}}{L}$,
      and so this case holds by induction.
    \item For case \rel{f-ref},
      $\forall a'$ where $a'\not=a, \Sigma(a') = \Sigma'(a')$. \\
      Since $\nvisible{\pc}{L}, L(\Sigma(a)) = 0$ by Lemma~\ref{lem:projLabel},
      as required.
    \item For case \rel{f-assign},
      $\forall a'$ where $a'\not=a, \Sigma(a') = \Sigma'(a')$. \\
      Since $\nvisible{\pc}{L}, L(\Sigma(a)) = L(\Sigma'(a))$ by Lemma~\ref{lem:projLabel},
      as required.
  \end{itemize}
\end{proof}
\noindent This lemma is also useful in the proof of the Projection Theorem.

\section{Proof of Theorem~\ref{thm:projection} (Projection)}
\label{sec:projection}
For convenience, we restate Theorem~\ref{thm:projection}. \\
{\it
\thmProjection
}
\noindent 
The proof extends $L$ to project evaluation contexts. They may project away the hole and so map evaluation contexts to expressions,
in which case filling the result is a no-op.

We capture in the following lemma the property that if a branch $B$ is inconsistent with the program counter $\pc$,
at most one of $B$ and $\pc$ may be visible to any given view $L$.
\begin{lem}
  \label{lem:inconBranches}
  If $B$ is inconsistent with $\pc$ and $\visible \pc L$,
  then $\nvisible{B}{L}$.
\end{lem}
\noindent With these properties established, we now prove projection.
\begin{proof}
By induction on the derivation of
  $\sstep {L(\Sigma)} {L(e)} {\emptyset} {L(\Sigma')} {L(V)}$
  and by case analysis on the final rule used in that derivation.
  \begin{itemize}
    \item The following cases hold trivially: \rel{f-val}, \rel{f-deref}, \rel{f-deref-null}, \rel{f-row}, \rel{f-project}, and \rel{f-union}.

    \item For case~\rel{f-select}, $e = \select{i}{j}{(\jtable{T})}$, so
      \[
        \sstep \Sigma {\select{i}{j}{(\jtable{T})}} {pc} \Sigma {(\jtable{T'})}
      \]
      where $T' = \{ (B,\overline{s}) ~|~ s_i=s_j \}$. \\
      This case holds
      since $L(\jtable T) = \{ (\emptyset, \overline{s}) ~|~ (B, \overline{s}) \in T, \visible{B}{L} \}$
      and $L(\jtable T') = \{ (\emptyset, \overline{s}) ~|~ (B, \overline{s}) \in T, \visible{B}{L}, s_i=s_j \}$, \\

    \item For case~\rel{f-join}, $e = \join{(\jtable{T_1})}{(\jtable{T_2})}$, so
      \[
        \sstep \Sigma {\join{(\jtable{T_1})}{(\jtable{T_2})}} {pc} \Sigma {(\jtable{T})}
      \]
      where $T = \{ B.B',\overline{s}.\overline{s'})
        ~|~ (B, \overline{s}) \in T_1, (B', \overline{s'}) \in T_2 \}$. \\
      $L(T) = \{ (B.B', \overline{s}.\overline{s'})
        ~|~ (B,\overline{s}) \in T_1, (B', \overline{s'}) \in T_2, \visible{B.B'}{L} \}$,
        so this case holds.

    \item For case~\rel{f-ctxt}, $e = E[e']$.
        By the antecedents of this rule
        \[
          \begin{array}{l}
            E \neq [] \\
            e' \mbox{ not a value} \\
            \sstep \Sigma {e'} {pc} {\Sigma_1} {V'} \\
            \sstep {\Sigma_1} {\ctxteval{E}{V'}} {pc} {\Sigma'} {V} \\
          \end{array}
        \]
        Note that $L(E[V']) = L(E)[L(V')]$, etc., so by induction
        \[
          \begin{array}{l}
            \sstep {L(\Sigma)} {L(e')} {\emptyset} {L(\Sigma_1)} {L(V')} \\
            \sstep {L(\Sigma_1)} {\ctxteval{L(E)}{L(V')}} {\emptyset} {L(\Sigma')} {L(V)} \\
          \end{array}
        \]
        Therefore,
        $\sstep {L(\Sigma)} {L(\ctxteval{E}{e})} {\emptyset} {L(\Sigma')} {L(V)}$, as required.

    \item For case \rel{f-strict}, $e = \ctxteval{S}{\labval{k}{V_1}{V_2}}$.
        By the antecedents of this rule
        \[
          \sstep \Sigma {\labval{k}{\ctxteval{S}{V_1}}{\ctxteval{S}{V_2}}} {pc} {\Sigma'} {V'}
        \]
        We now consider each possible case for the next step in the derivation.
        \begin{itemize}
          \item For subcase~\rel{f-left}, we know that $k\in\pc, k\in L$ and
            \[
              \sstep \Sigma {\ctxteval{S}{V_1}} {\emptyset} {\Sigma'} {V}
            \]
            By induction,
          $\sstep {L(\Sigma)} {L(\labval{k}{\ctxteval{S}{V_1}}{\ctxteval{S}{V_2}})} {\emptyset} {L(\Sigma')} {L(V')}$. \\
          \item Subcase~\rel{f-right} holds by a similar argument.
          \item For subcase~\rel{f-split}, $k\not\in\pc, \notlab k\not\in\pc$ and
            \[
              \begin{array}{l}
                \sstep{\Sigma} {\ctxteval{S}{V_1}} {\unionk{\pc}{k}} {\Sigma''} {V''} \\
                \sstep{\Sigma''} {\ctxteval{S}{V_2}} {\unionk{\pc}{\notlab k}} {\Sigma'} {V'''} \\
                V = \labResult{k}{V''}{V'''}
              \end{array}
            \]
            \begin{itemize}
              \item
                If $k\in L$, then by induction we have
                $\sstep{L(\Sigma)} {L(\ctxteval{S}{V_1})} {\emptyset} {L(\Sigma'')} {L(V'')}$. \\
                $L(\Sigma'') = L(\Sigma')$ by Lemma~\ref{lem:invis},
                and $L(V) = L(V'')$. \\
                Therefore,
                $\sstep{L(\Sigma)} {L(\ctxteval{S}{V_1})} {\emptyset} {L(\Sigma')} {L(V')}$,
                as required.
              \item
                If $k\not\in L$, then this case holds by a similar argument.
            \end{itemize}
        \end{itemize}

    \item For case \rel{f-fold-empty}, we have
        $$\sstep {\Sigma} {\fold{V_f}{V_b}{(\jtable{\epsilon})}} {\pc} {\Sigma} {V_b}$$
        Clearly,
        $\sstep {L(\Sigma)} {\fold{L(V_f)}{L(V_b)}{L(\jtable{\epsilon})}} {\emptyset} {L(\Sigma)} {L(V_b)}$. \\

    \item For case \rel{f-fold-inconsistent}, we have
      $$e = {\fold{V_f}{V_p}{(\jtable{(B, \jrow{s}).T})}}$$.
        By the antecedents of this rule, we have
        $$
        \begin{array}{c}
          \sstep {\Sigma} {\fold{V_f}{V_b}{(\jtable{T})}} {\pc} {\Sigma'} {V} \\
          B \t{ is inconsistent with } \pc \\
        \end{array}
        $$
        By Lemma~\ref{lem:inconBranches}, $\nvisible{B}{L}$. \\
        Therefore, $L(\jtable{(B, \overline{s}).T}) = L(\jtable T)$. \\
        By the \rel{f-fold-empty} rule,
        \[
          \sstep {L(\Sigma)} {\fold{L(V_f)}{L(V_b)}{L(\jtable{(B, \overline{s}).T})}} {\emptyset} {L(\Sigma')} {L(V)}
        \]
        By induction,
        $\sstep {L(\Sigma)} {L(\fold{V_f}{V_b}{(\jtable{T})})} {\emptyset} {L(\Sigma')} {L(V)}$,
        as required.

      \item
        For case \rel{f-fold-consistent}, $e = \fold{V_f}{V_b}{(\jtable{T})}$. \\
        By the antecedents of this rule, we have
        \[
          \begin{array}{c}
            \sstep {\Sigma} {\fold{V_f}{V_b}{(\jtable{T})}} {\pc} {\Sigma_1} {V_1} \\
            B \t{ is consistent with } \pc \\
            \sstep {\Sigma_1} {V_f~\jrow{s}~V_1} {pc \cup B} {\Sigma'} {V_2} \\
            V = \labResult{B}{V_2}{V_1}
          \end{array}
        \]
        \begin{itemize}
          \item If $\visible B L$, then $\visible{\union{\pc}{B}}{L}$. \\
            By induction, \[
              \begin{array}{c}
                \sstep {L(\Sigma)} {L(\fold{V_f}{V_b}{(\jtable{T})})} {\emptyset} {L(\Sigma_1)} {L(V_1)} \\
                \sstep {L(\Sigma_1)} {L(V_f~\jrow{s}~V_1)} {\emptyset} {L(\Sigma')} {L(V_2)} \\
              \end{array}
            \]
            By Lemma~\ref{lem:projBranch}, $L(V) = L(\labResult{B}{V_2}{V_1})$, as required.
          \item Otherwise, $\nvisible{B}{L}$, and therefore $\nvisible{\union{\pc}{B}}{L}$.
            By Lemma~\ref{lem:invis}, $L(\Sigma_1) = L(\Sigma')$. \\
            We have
            $\sstep {L(\Sigma)} {L(\fold{V_f}{V_b}{(\jtable{T})})} {\emptyset} {L(\Sigma_1)} {L(V_1)}$ by induction. \\
            $L(\jtable{(B,\overline{s}).T}) = L(\jtable{T})$. \\
            By Lemma~\ref{lem:projBranch}, $L(V) = L(\labResult{B}{V_2}{V_1})$, as required.
        \end{itemize}

      \item
        For case~\rel{f-left}, $e = \labexpr{k}{e_1}{e_2}$. \\
        By the antecedents of this rule, we have
        \[
          \begin{array}{c}
            k \in\pc \\
            \sstep \Sigma {e_1} \pc {\Sigma'} {V} \\
          \end{array}
        \]
        Since $k\in\pc$, $L(e) = L(e_1)$. \\
        By induction, $\sstep {L(\Sigma)} {L(e_1)} \emptyset {L(\Sigma')} {L(V)}$. \\

      \item Case~\rel{f-right} holds by a similar argument.

      \item For case~\rel{f-split}, $e = \labexpr{k}{e_1}{e_2}$. \\
        By the antecedents of this rule, we have
        \[
          \begin{array}{c}
            k \not\in\pc \qquad \notlab k \not\in\pc \\
            \sstep \Sigma {e_1} {\unionk \pc k} {\Sigma_1} {V_1} \\
            \sstep {\Sigma_1} {e_2} {\unionk \pc {\notlab k}} {\Sigma'} {V_2} \\
            V = \labResult{k}{V_1}{V_2}
          \end{array}
        \]
        \begin{itemize}
          \item If $k\in L$, then by induction
            $\sstep {L(\Sigma)} {L(e_1)} \emptyset {L(\Sigma_1)} {L(V_1)}$. \\
            $L(\Sigma_1) = L(\Sigma')$ by Lemma~\ref{lem:invis},
            and by Lemma~\ref{lem:projLabel} \\
            $L(V) = L(\labResult{k}{V_1}{V_2}) = L(V_1)$, as required.
          \item Otherwise $\notlab k\in L$, so
            $L(\Sigma) = L(\Sigma_1)$ by Lemma~\ref{lem:invis}. \\
            By induction, $\sstep {L(\Sigma_1)} {L(e_2)} \emptyset {L(\Sigma')} {L(V_2)}$, \\
            and by Lemma~\ref{lem:projLabel}
            $L(V) = L(\labResult{k}{V_1}{V_2}) = L(V_2)$, as required.
        \end{itemize}

      \item For case~\rel{f-app}, $e = (\lam{x}{e'}~V')$.
        By the rule antecedents,
        \[
            \sstep \Sigma {e'[x:=V']} \pc {\Sigma'} {V} \\
        \]
        We know that $L(e) = L(\lam x {e'}~V') = L(e'[x:=V'])$. \\
        By induction,
        $\sstep {L(\Sigma)} {L({e'[x:=V']})} \emptyset {L(\Sigma')} {L(V)}$, as required.

      \item For case~\rel{f-ref}, $e = \mkref V'$.
        By the rule antecedents
        \[
          \begin{array}{c}
            a \not\in \fun{dom}(\Sigma) \\
            \Sigma' = \Sigma[a := \labResult{\pc}{V'}{0}] \\
          \end{array}
        \]
        Without loss of generality, we assume that both evaluations
        allocate the same address $a$.
        Since $a\not\in\fun{dom}(\Sigma), a\not\in\fun{dom}(L(\Sigma))$. \\
        Also, we know that $\forall a'\in\fun{dom}(\Sigma), \Sigma(a') = \Sigma'(a')$,
        and therefore $L(\Sigma(a')) = L(\Sigma'(a'))$. \\
        Since $\visible{\pc}{L}$, $L(\Sigma'(a)) = L(\labResult{\pc}{V'}{0}) = L(V')$
        by Lemma~\ref{lem:projBranch}.
        Since $L(\labResult{\emptyset}{V'}{0}) = L(V') = L(V)$,
        this case holds.

      \item For case~\rel{f-assign}, $e = (\assign{a}{V})$.
        By the antecedent of this rule,
        $\Sigma' = \Sigma[a := \labResult{\pc}{V}{\Sigma(a)}]$.
        We know $\forall a'\in\fun{dom}(\Sigma), \Sigma(a') = \Sigma'(a')$,
        and therefore $L(\Sigma(a')) = L(\Sigma'(a'))$. \\
        Since $\visible L \pc$,
        $L(\Sigma'(a)) = L(\labResult{\pc}{V}{\Sigma(a)}) = L(V)$
        by Lemma~\ref{lem:projBranch}.
        And since $L(\labResult{\emptyset}{V}{\Sigma(a)}) = L(V)$,
        this case holds.
  \end{itemize}
\end{proof}

\end{document}